\documentclass[journal,twoside,web]{ieeecolor}
\usepackage{generic}
\usepackage{cite}
\usepackage{amsmath,amssymb,amsfonts}
\usepackage{algorithmic}
\usepackage{graphicx, subfigure}
\usepackage{textcomp}
\usepackage{bm}
\usepackage{float}
\usepackage{lipsum}
\usepackage{mathrsfs}
\usepackage{makecell}
\usepackage{slashbox}
\usepackage{dblfloatfix}
\usepackage{commath}
\usepackage{amstext}

\definecolor{green}{rgb}{0.0, 0.5, 0.0}

\newtheorem{theorem}{Theorem}[section]

\newtheorem{lemma}[theorem]{Lemma}

\newtheorem{assumption}{Assumption}[section]

\DeclareGraphicsExtensions{.eps}

\graphicspath{{images/}}

\def\BibTeX{{\rm B\kern-.05em{\sc i\kern-.025em b}\kern-.08em
    T\kern-.1667em\lower.7ex\hbox{E}\kern-.125emX}}
\begin{document}
\title{Information-Weighted Consensus Filter with Partial Information Exchange}
\author{Byoung-Ju Jeon and Shaoming He\textsuperscript{*}
\thanks{Byoung-Ju Jeon is with School of Aerospace, Transport and Manufacturing (SATM), Cranfield University, College Road, Cranfield MK430AL, United Kingdom (e-mail: B.Jeon@cranfield.ac.uk). }
\thanks{Shaoming He is with the School of Aerospace Engineering, Beijing Institute of Technology, Beijing 100081, China (e-mail:shaoming.he@bit.edu.cn).}
\thanks{\textsuperscript{*}Corresponding Author. Email: $\texttt{shaoming.he@bit.edu.cn}$}
}

\maketitle

\begin{abstract}
In this note, the information-weighted consensus filter (ICF) with partial information exchange is proposed to reduce the bandwidth of the signals transmitted between the sensor nodes and guarantee its convergence to the centralized Kalman filter (CKF).
In the proposed algorithm, a part of information chosen with the entry selection matrix is transmitted to the sensor nodes in the neighborhood at each consensus step, and consensus averaging is conducted at each sensor node with the partial and the local information.
This ensures that the proposed distributed estimation algorithm converges to the centralized algorithm, while allowing the proposed algorithm to achieve bandwidth reduction of the signals transmitted between the sensors.
With the proposed algorithm, the stability of the estimation error dynamics is proven and the convergence to the centralized algorithm is mathematically shown using the property of the average consensus.
Simulations are conducted to validate the proposed ICF with partial information exchange and the related theoretical findings.
\end{abstract}


\section{Introduction}
\label{sec:introduction}
\IEEEPARstart{D}{istributed} estimation algorithms have been widely studied for target tracking\cite{li2019comprehensive,li2017distributed,talebi2019distributed,olfati2005distributed,tamjidi2021unifying} due to their scalability to large networks and high fault tolerance \cite{kamal2013information,cattivelli2010diffusion,he2020distributed}.
Distributed algorithms perform estimation at each sensor node using local information from the sensor nodes in their neighborhoods, while centralized algorithms conduct estimation at a fusion center using information from all sensor nodes.
Despite the advantages of the distributed estimation algorithms, there exist several technical issues to be tackled for their successful development and application.
One of the critical issues is that the bandwidth of the signals transmitted between the sensor nodes for the distributed estimation is limited.
Another main issue is that performance of the distributed estimation algorithms can be degenerated comparing to the centralized estimation algorithms.

To resolve the issue arising from the bandwidth limit, previous studies have focused on reducing the amount of information transmitted.
The approach suggested in \cite{ribeiro2006bandwidth} is to quantize binary observations from sensors and transmit only one or a few bits per observation.
In \cite{orihuela2018negotiated}, an iterative procedure similar to a non-cooperative game trades off communication burden and estimation accuracy through the negotiation between sensor nodes.
In \cite{schizas2007distributed}, sensor data transmitted are compressed via dimensionality reduction.
The partial diffusion algorithm proposed in \cite{vahidpour2019performance, vahidpour2019partial} reduces the amount of information exchanged between the sensor nodes by selecting the entries to be transmitted simply using an entry selection matrix without complex algorithms.

To maintain accuracy of the state estimates without the fusion center, previous studies suggested distributed algorithms which aim to guarantee convergence to the centralized algorithms.
In \cite{chiuso2011gossip}, a distributed gossip-based algorithm is designed to generate the maximum-likelihood estimates and this algorithm converges to the centralized algorithm as the number of the sensor nodes goes to infinity.
Distributed adaptive algorithms in \cite{xu2012distributed} utilize a conjugate gradient, and their performances were shown to be similar with the centralized algorithms through simulations, not mathematical proofs.
The information weighted consensus filter (ICF) in \cite{kamal2012information,kamal2013information} takes weighted average on the information from the sensor nodes in the neighborhood, and this algorithm is guaranteed to asymptotically converge to the centralized Kalman filter (CKF) using the property of average consensus. Later in \cite{battistelli2014consensus,battistelli2016stability,wang2017convergence}, the ICF algorithm is proven to be a special case of the hybrid consensus on measurement and consensus on information (HCMCI) approach and the estimation error boundedness is also theoretically ensured.

In this note, the ICF with partial information exchange is proposed to accomplish both bandwidth reduction of the signals transmitted and guaranteed convergence to the CKF.
To the best of our knowledge, there have been no distributed estimation algorithms which achieve both objectives simultaneously.
For instance, the partial diffusion Kalman filtering in \cite{vahidpour2019partial} reduces the amount of information transmitted but its convergence to the CKF is not guaranteed, while ICF in \cite{kamal2013information} is guaranteed to converge to the CKF but all information should be transmitted between the sensor nodes.
In the proposed algorithm, consensus averaging of the local information is performed at each sensor node, and this guarantees the convergence to the CKF.
A part of the information to be transmitted at each consensus step is selected with the entry selection matrix to reduce the bandwidth of the signals transmitted between the sensor nodes in the neighborhood.
It is proven that the estimation error from the proposed ICF with partial information exchange is stable, and convergence of the proposed algorithm to the CKF is also mathematically shown to be guaranteed.

The rest of this note is organized as follows.
The problem formulation for the distributed estimation is defined in Section II.
The design of the ICF with partial information exchange is proposed in Section III.
The stability of the estimation error dynamics with the proposed algorithm is proven in Section IV.
In Section V, it is proven that the ICF with partial information exchange converges to CKF.
Simulation results are provided in Section VI and overall concluding remarks are given in Section VII.

\section{Problem Formulation}
A sensor network can be expressed as an undirected connected graph $\mathcal{G}=\left( \mathcal{N}, \mathcal{A} \right)$.
$\mathcal{N} = \left\{ 1,2,\cdots,N \right\}$ is the set of vertices of the graph $\mathcal{G}$, and the vertices represent the sensor nodes.
$\mathcal{A}$ is the set of edges of the graph $\mathcal{G}$, and the edge $\left( i, j \right) \in \mathcal{A}$ indicates that the sensor node $i$ can receive information from $j$.
It is assumed that only single-hop communications are available.
This means that the sensor node $i$ can communicate only with the sensor nodes in the set of its neighbors $\mathcal{N}^{i} = \left\{ j | \left(i ,j \right) \in \mathcal{A} \right\}$.
Note that $\mathcal{N}^{i}$ includes $i$ itself.

The system dynamics and the measurement equation are defined as
\begin{equation}
\begin{split} \label{eq:1}
\bm{x}_{t+1}=\bm{f}_{t}(\bm{x}_{t})+\bm{w}_{t} \\
\bm{y}_{t}^{i}=\bm{h}_{t}^{i}(\bm{x}_{t})+\bm{v}_{t}^{i}
\end{split}
\end{equation}
where $\bm{x}_{t} \in \mathbb{R}^{n \times 1}$ is the state vector and $\bm {y}_{t}^{i} \in \mathbb{R}^{m^{i} \times 1}$ is the measurement vector obtained by the sensor node $i$. 
Note that $\bm{x}_{t}$ is assumed to be observable from $\left\{ \bm {y}_{t}^{i} \right\}_{i=1}^{N}$.
$\bm{f}_{t} \in \mathbb{R}^{n \times 1}$ and $\bm{h}_{t}^{i} \in \mathbb{R}^{m^{i} \times 1}$ are vectors of nonlinear functions. 
$\bm{w}_{t} \in \mathbb{R}^{n \times 1}$ represents the process noise and $\bm{v}_{t}^{i} \in \mathbb{R}^{m^{i} \times 1}$ denotes the measurement noise generated by the sensor node $i$.
$\bm{w}_{t}$ and $\bm{v}_{t}^{i}$ are modeled as zero-mean Gaussian noises and their covariance matrices are defined as
\begin{equation}
\begin{split} \label{eq:2}
&Q_{t}=E[\bm{w}_{t} \bm{w}_{t}^{T}]>0 \\
&R_{t}^{i}=E[\bm{v}_{t}^{i} (\bm{v}_{t}^{i})^{T}]>0
\end{split}
\end{equation}
where $Q_{t} \in \mathbb{R}^{n \times n}$ is the covariance matrix of $\bm{w}_{t}$ and $R_{t}^{i} \in \mathbb{R}^{m^{i} \times m^{i}}$ is the covariance matrix of $\bm{v}_{t}^{i}$.

\section{ICF with Partial Information Exchange}

\subsection{Benchmark : C(E)KF}

As a benchmark for the ICF with partial information exchange, a centralized (Extended) Kalman Filter (C(E)KF) in the information filter form is provided in Table \ref{table_1} \cite{battistelli2014consensus}.
\begin{table}[thb]
\caption{C(E)KF in a Information Filter Form}
\label{table_1}
\begin{center}
\begin{tabular}{ | l |}
\hline
\textbf{Correction (Measurement-Update):} \\
$C_{t}^{i}=\frac{\partial \bm{h}_{t}^{i}}{\partial \bm{x}_{t}} \left( \hat{\bm{x}}_{t|t-1} \right)$ \\
${\Omega}_{t|t} = {\Omega}_{t|t-1} + \sum_{i} \left( {C}_{t}^{i} \right)^{T} {V}_{t}^{i} {C}_{t}^{i}$ \\
$\bar{\bm{y}}_{t}^{i} = \bm{y}_{t}^{i} - \bm{h}_{t}^{i} \left( \hat{\bm{x}}_{t|t-1} \right) + C_{t}^{i} \hat{\bm{x}}_{t|t-1}$\\
$\bm{q}_{t|t} = \bm{q}_{t|t-1} + \sum_{i} \left( {C}_{t}^{i} \right)^{T} {V}_{t}^{i} \bar{\bm{y}}_{t}^{i}$ \\

\textbf{Prediction (Time-Update) :} \\
$\hat{\bm{x}}_{t|t}={\Omega}_{t|t}^{-1} \bm{q}_{t|t}$ \\
$A_{t}=\frac{\partial \bm{f}_{t}}{\partial \bm{x}_{t}} \left( \hat{\bm{x}}_{t|t} \right)$ \\
${\Omega}_{t+1|t}=\left( A_{t} {\Omega}_{t|t}^{-1}A_{t}^{T}+W_{t}^{-1} \right)^{-1}$ \\
$\hat{\bm{x}}_{t+1|t}=\bm{f}_{t} \left( \hat{\bm{x}}_{t|t} \right)$ \\
$\bm{q}_{t+1|t} = {\Omega}_{t+1|t} \hat{\bm{x}}_{t+1|t} $ \\
\hline
\end{tabular}
\end{center}
\end{table}

The information matrices are defined as 
\begin{equation}
\begin{split} \label{eq:3}
\Omega_{t|t-1} \triangleq P_{t|t-1}^{-1} \qquad \qquad \Omega_{t|t} \triangleq P_{t|t}^{-1}
\end{split}
\end{equation}
where $P_{t|t-1} \in \mathbb{R}^{n \times n}$ and $P_{t|t} \in \mathbb{R}^{n \times n}$ are a priori and a posteriori covariance matrices, respectively.
The information vectors are defined as 
\begin{equation}
\begin{split} \label{eq:4}
\bm{q}_{t|t-1} \triangleq P_{t|t-1}^{-1} \hat{\bm{x}}_{t|t-1} \qquad \bm{q}_{t|t} \triangleq P_{t|t}^{-1} \hat{\bm{x}}_{t|t}
\end{split}
\end{equation}
where $\hat{\bm{x}}_{t|t-1}$ and $\hat{\bm{x}}_{t|t}$ are a priori and a posteriori state estimate vectors.
With the covariance matrices of the process noise and the measurement noise, noise information matrices $W_{t}$ and $V_{t}^{i}$ are defined as
\begin{equation}\label{eq:5}
W_{t} \triangleq Q_{t}^{-1},\quad V_{t}^{i} \triangleq (R_{t}^{i})^{-1}
\end{equation}

\subsection{Design of ICF with Partial Information Exchange}

In this note, the ICF with partial information exchange is proposed to reduce the amount of information transmitted between the sensor nodes while guaranteeing the convergence to the centralized algorithm. The pseudo-code of the proposed algorithm is summarized in Table \ref{table_2}.

\begin{table}[thb]
\caption{ICF with Partial Information Exchange}
\label{table_2}
\begin{center}
\begin{tabular}{ | l |}
\hline
\textbf{Compute the local correction terms} \\
Sample the measurement $\bm{y}_{t}^{i}$ \\
$\delta \bm{q}_{t}^{i} = \left( {C}_{t}^{i} \right)^{T} {V}_{t}^{i} \bm{y}_{t}^{i}$ \\
$\delta {\Omega}_{t}^{i} = \left( {C}_{t}^{i} \right)^{T} {V}_{t}^{i} {C}_{t}^{i}$ \\
\textbf{Consensus :} \\
${B}_{t}^{i} \left( 0 \right) = \frac{1}{N}{\Omega}_{t|t-1}^{i} + \delta {\Omega}_{t}
^{i}$ \\
$\bm{b}_{t}^{i} \left( 0 \right) = \frac{1}{N}\bm{q}_{t|t-1}^{i} + \delta \bm{q}_{t}^{i}$ \\
for $l = 0, 1, \dots, L-1$ ($L = c {\bar{\theta}}$ where $c$ : finite integer) do \\
$\qquad$ $ {B}_{t}^{i} \left( l+1 \right) = {B}_{t}^{i} \left( l \right) + \epsilon \sum_{j \in \mathcal{N}^{i}} T_{t, l}^{j} \left[ {B}_{t}^{j} \left( l \right) - {B}_{t}^{i} \left( l \right) \right]$ \\
$\qquad$ $ \bm{b}_{t}^{i} \left( l+1 \right) = \bm{b}_{t}^{i} \left( l \right) + \epsilon \sum_{j \in \mathcal{N}^{i}} T_{t, l}^{j} \left[ \bm{b}_{t}^{j} \left( l \right) - \bm{b}_{t}^{i} \left( l \right) \right]$ \\
end for \\

\textbf{Correction :} \\
${\Omega}_{t|t}^{i} = N {B}_{t}^{i} \left( L \right)$ \\
$\bm{x}_{t|t}^{i} = \left( {B}_{t}^{i} \left( L \right) \right)^{-1} \bm{b}_{t}^{i} \left( L \right)$ \\

\textbf{Prediction :} \\
$\hat{\bm{x}}_{t+1|t}^{i}=\bm{f}_{t}(\hat{\bm{x}}_{t|t}^{i})$, and $A_{t}=\frac{\partial \bm{f}_{t}}{\partial \bm{x}_{t}}(\hat{\bm{x}}_{t|t}^{i})$ \\
${\Omega}_{t+1|t}^{i}=(A_{t} {{\Omega}_{t|t}^{i}}^{-1}A_{t}^{T}+W_{t}^{-1})^{-1}$ \\
$\bm{q}_{t+1|t}^{i}={\Omega}_{t+1|t}^{i}\hat{\bm{x}}_{t+1|t}^{i}$ \\
\hline
\end{tabular}
\end{center}
\end{table}

For the bandwidth reduction, the data to be transmitted from the sensor node $j \in \mathcal{N}^{i}$ to the sensor node $i$ are selected with the entry selection matrix $T_{t,l}^{i}$ \cite{vahidpour2019partial} in \eqref{eq:8}.
\begin{equation}
\begin{split} \label{eq:8}
&\qquad \qquad \ \ T_{t, l}^{i} = \left[ \begin{matrix}
\tau_{1, t, l}^{i} & \cdots & 0\\
\vdots & \ddots & \vdots \\
0 & \cdots & \tau_{n, t, l}^{i}
\end{matrix} \right] \\
&\mbox{where}\quad \tau_{r, t, l}^{i} = 
\begin{cases}
1 & \mbox{if} \ r \in J_{(l \ \mbox{mod} \ \bar{\theta})} \\
0 & \mbox{otherwise}
\end{cases}, \quad \bar{\theta}=\lceil n / m \rceil
\end{split}
\end{equation} 
$m$ is the number of the selected entries at each iteration, and it is restricted by $m \leq n$. $\bar{\theta}$ is the number of different entry selection matrices and is smaller or equal to the number of consensus step $L$.
The coefficient subsets $J_{\zeta}$ should satisfy the following three conditions:
(1) the cardinality of $J_{\zeta}$ is between $1$ and $m$;
(2) $\cup_{\zeta=1}^{\bar{\theta}} J_{\zeta}=S$, where $S=\{1,2,\cdots,n\}$;
and (3) $J_{\zeta_{1}} \cap J_{\zeta_{2}} = \emptyset$, $\forall \zeta_{1}, \zeta_{2} \in \{1,\cdots,\bar{\theta}\}$ and $\zeta_{1} \neq \zeta_{2}$.

The convergence to the CKF is guaranteed by performing consensus averaging at each sensor node $i$ with information transmitted from the sensor nodes in the neighborhood as
\begin{equation}
\begin{split} \label{eq:7}
{B}_{t}^{i} \left( l+1 \right) = {B}_{t}^{i} \left( l \right) + \epsilon \sum_{j \in \mathcal{N}^{i}} T_{t, l}^{j} \left[ {B}_{t}^{j} \left( l \right) - {B}_{t}^{i} \left( l \right) \right] \\
\bm{b}_{t}^{i} \left( l+1 \right) = \bm{b}_{t}^{i} \left( l \right) + \epsilon \sum_{j \in \mathcal{N}^{i}} T_{t, l}^{j} \left[ \bm{b}_{t}^{j} \left( l \right) - \bm{b}_{t}^{i} \left( l \right) \right]
\end{split}
\end{equation}
where $\epsilon>0$ is the consensus gain. Note that the proposed algorithm reduces to the original ICF if an identity entry selection matrix is utilized.

\section{Stability Analysis}
The stability of the proposed ICF with partial information exchange is proven with a linear time-invariant (LTI) system.
Hence, the system dynamics and the measurement equation are supposed to have the following form.
\begin{equation}
\begin{split} \label{eq:9}
\bm{x}_{t+1}=A\bm{x}_{t}+\bm{w}_{t} \\
\bm{y}_{t}^{i}=C^{i}\bm{x}_{t}+\bm{v}_{t}^{i}
\end{split}
\end{equation}

The process and the measurement noise covariance matrices are assumed to be time-invariant, i.e., $Q_{t}=Q>0$ and $R_{t}^{i}=R^{i}>0$. The assumptions required for the stability analysis are suggested in Assumption \ref{assumption_1}.
\begin{assumption}
\label{assumption_1}
(Assumptions for Stability Analysis) \\
A1. The system matrix $A$ is invertible. \\
A2. The system is collectively observable. \\
A3. The consensus matrix $\Pi$, whose elements are the consensus weights $\epsilon$, is row stochastic and primitive. \\
\end{assumption}

In order to discuss the stability of the proposed ICF with partial information exchange, the first step is to investigate the boundedness of the information matrix $\Omega^{i}$.
\begin{lemma}
\label{lemma_1}
(Boundedness of the Information Matrix)
Let the Assumption \ref{assumption_1} hold. Then, there exist positive definite matrices $\underline{\Omega}$, $\bar{\Omega}$, $\underline{\Omega}^{+}$, and $\bar{\Omega}^{+}$ such that $0 \leq \underline{\Omega} \leq \Omega_{t|t}^{i} \leq \bar{\Omega}$ and $0 \leq \underline{\Omega}^{+} \leq \Omega_{t+1|t}^{i} \leq \bar{\Omega}^{+}$ for any $i \in {\mathcal N}$ and $t \geq 1$.
\end{lemma}

\begin{proof}
From Table \ref{table_2}, ${\Omega}_{t|t}^{i}$ can be rewritten as
\begin{equation}
\resizebox{0.43\textwidth}{!}{$
\begin{split} \label{eq:10}
{\Omega}_{t|t}^{i} =& N {B}_{t}^{i} \left( L \right)\\
=& N \left[ \sum_{j\in\mathcal{N}_{i}} \epsilon^{L} \left\{ \sum_{k=1}^{L} T_{t,k-1}^{j} \prod_{m=1}^{k-1} \left( I -T_{t,m-1}^{j} \right) B_{t}^{j} \left( 0 \right) \right\} \right] \\
=& \left[ \sum_{j\in\mathcal{N}_{i}} \epsilon^{L} \left\{ \sum_{k=1}^{L} T_{t,k-1}^{j} \prod_{m=1}^{k-1} \left( I -T_{t,m-1}^{j} \right) {\Omega}_{t|t-1}^{j} \right\} \right] \\
&+N \left[ \sum_{j\in\mathcal{N}_{i}} \epsilon^{L} \left\{ \sum_{k=1}^{L} T_{t,k-1}^{j} \prod_{m=1}^{k-1} \left( I -T_{t,m-1}^{j} \right) \right. \right. \\& \left. \left. \times\left( {C}_{t}^{j} \right)^{T} {V}_{t}^{j} {C}_{t}^{j} \right\} \right]
\end{split}
$}
\end{equation}
If two positive scalars $\underline{\omega}$ and $\bar{\omega}$ which satisfy $0<\underline{\omega}<N<\bar{\omega}$ are defined, the upper and the lower bounds of ${\Omega}_{t|t}^{i}$ can be derived as \eqref{eq:11} and \eqref{eq:12} from \eqref{eq:10}.
\begin{equation}
\resizebox{0.43\textwidth}{!}{$
\begin{split} \label{eq:11}
{\Omega}_{t|t}^{i} \leq &\left[ \sum_{j\in\mathcal{N}_{i}} \epsilon^{L} \left\{ \sum_{k=1}^{L} T_{t,k-1}^{j} \prod_{m=1}^{k-1} \left( I -T_{t,m-1}^{j} \right) {\Omega}_{t|t-1}^{j} \right\} \right] \\
&+ \bar{\omega} \left[ \sum_{j\in\mathcal{N}_{i}} \epsilon^{L} \left\{ \sum_{k=1}^{L} T_{t,k-1}^{j} \prod_{m=1}^{k-1} \left( I -T_{t,m-1}^{j} \right) \right. \right. \\& \left. \left. \times\left( {C}_{t}^{j} \right)^{T} {V}_{t}^{j} {C}_{t}^{j} \right\} \right]
\end{split}
$}
\end{equation}
\begin{equation}
\resizebox{0.43\textwidth}{!}{$
\begin{split} \label{eq:12}
{\Omega}_{t|t}^{i} \geq &\left[ \sum_{j\in\mathcal{N}_{i}} \epsilon^{L} \left\{ \sum_{k=1}^{L} T_{t,k-1}^{j} \prod_{m=1}^{k-1} \left( I -T_{t,m-1}^{j} \right) {\Omega}_{t|t-1}^{j} \right\} \right] \\
&+ \underline{\omega} \left[ \sum_{j\in\mathcal{N}_{i}} \epsilon^{L} \left\{ \sum_{k=1}^{L} T_{t,k-1}^{j} \prod_{m=1}^{k-1} \left( I -T_{t,m-1}^{j} \right) \right. \right. \\& \left. \left. \times\left( {C}_{t}^{j} \right)^{T} {V}_{t}^{j} {C}_{t}^{j} \right\} \right]
\end{split}
$}
\end{equation}
Suppose that $\bar{\Omega}_{t|t}^{i}$ and $\bar{\Omega}_{t|t-1}^{i}$ are the information matrices generated by replacing $N$ with $\bar{\omega}$ in the ICF with partial information exchange.
Then, $\bar{\Omega}_{t|t}^{i}$ can be further formulated using $\bar{\Omega}_{t|t-1}^{i}$ as
\begin{equation}
\resizebox{0.43\textwidth}{!}{$
\begin{split} \label{eq:13}
\bar{\Omega}_{t|t}^{i} =& \left[ \sum_{j\in\mathcal{N}_{i}} \epsilon^{L} \left\{ \sum_{k=1}^{L} T_{t,k-1}^{j} \prod_{m=1}^{k-1} \left( I -T_{t,m-1}^{j} \right) \bar{\Omega}_{t|t-1}^{j} \right\} \right] \\
&+ \bar{\omega} \left[ \sum_{j\in\mathcal{N}_{i}} \epsilon^{L} \left\{ \sum_{k=1}^{L} T_{t,k-1}^{j} \prod_{m=1}^{k-1} \left( I -T_{t,m-1}^{j} \right) \right. \right. \\& \left. \left.\times \left( {C}_{t}^{j} \right)^{T} {V}_{t}^{j} {C}_{t}^{j} \right\} \right]
\end{split}
$}
\end{equation}
Similarly, suppose that $\underline{\Omega}_{t|t}^{i}$ and $\underline{\Omega}_{t|t-1}^{i}$ are the information matrices generated by replacing $N$ with $\underline{\omega}$ in the ICF with partial information exchange.
Then, $\bar{\Omega}_{t|t}^{i}$ can be derived from $\bar{\Omega}_{t|t-1}^{i}$ as
\begin{equation}
\resizebox{0.43\textwidth}{!}{$
\begin{split} \label{eq:14}
\underline{\Omega}_{t|t}^{i} = &\left[ \sum_{j\in\mathcal{N}_{i}} \epsilon^{L} \left\{ \sum_{k=1}^{L} T_{t,k-1}^{j} \prod_{m=1}^{k-1} \left( I -T_{t,m-1}^{j} \right) \underline{\Omega}_{t|t-1}^{j} \right\} \right] \\
&+ \underline{\omega} \left[ \sum_{j\in\mathcal{N}_{i}} \epsilon^{L} \left\{ \sum_{k=1}^{L} T_{t,k-1}^{j} \prod_{m=1}^{k-1} \left( I -T_{t,m-1}^{j} \right) \right. \right. \\& \left. \left. \times\left( {C}_{t}^{j} \right)^{T} {V}_{t}^{j} {C}_{t}^{j} \right\} \right]
\end{split}
$}
\end{equation}
If $\underline{\Omega}_{t|t}^{i}$, ${\Omega}_{t|t}^{i}$ and $\bar{\Omega}_{t|t}^{i}$ are initialized with the same positive definite matrix ${\Omega}_{0|0}^{i}$, we have
\begin{equation}
\begin{split} \label{eq:15}
\underline{\Omega}_{t|t}^{i} \leq {\Omega}_{t|t}^{i} \leq \bar{\Omega}_{t|t}^{i}
\end{split}
\end{equation}
Since the prediction in the recursive steps of the Kalman filter is a monotonically nondecreasing process, \eqref{eq:15} implies that
\begin{equation}
\begin{split} \label{eq:16}
\underline{\Omega}_{t+1|t}^{i} \leq {\Omega}_{t+1|t}^{i} \leq \bar{\Omega}_{t+1|t}^{i}
\end{split}
\end{equation}
Under the Assumption \ref{assumption_1} and the fact that ${\Omega}_{1|0}^{i}>0$, $\underline{\Omega}_{t|t}^{i}$, $\underline{\Omega}_{t+1|t}^{i}$, $\bar{\Omega}_{t|t}^{i}$ and $\bar{\Omega}_{t+1|t}^{i}$ are positive definite and increasing over finite time.
This implies that there exist positive definite $\underline{\Omega}$, $\bar{\Omega}$, $\underline{\Omega}^{+}$ and $\bar{\Omega}^{+}$, which satisfy $\underline{\Omega}_{t|t}^{i} \geq \underline{\Omega}$, $\bar{\Omega}_{t|t}^{i} \leq \bar{\Omega}$, $\underline{\Omega}_{t+1|t}^{i} \geq \underline{\Omega}^{+}$ and $\bar{\Omega}_{t+1|t}^{i} \leq \bar{\Omega}^{+}$.
Thus, Lemma \ref{lemma_1} holds.
\end{proof}

As the second step of the stability analysis on the proposed algorithm, the dynamics of the estimation error is derived and suggested in Lemma \ref{lemma_2}.
\begin{lemma}
\label{lemma_2}
(Estimation Error Dynamics)
Let the Assumption \ref{assumption_1} hold. Then, the estimation error $\bm{e}_{t}^{i}$ follows the dynamic equation as
\begin{equation}
\begin{split} \label{eq:17}
&\qquad \quad \bm{e}_{t+1}^{i}=\sum_{j \in \mathcal{N}} \Phi_{t}^{i,j} \bm{e}_{t}^{j} + \sum_{j \in \mathcal{N}} \Gamma_{t}^{i,j} \bm{v}_{t}^{j} + \bm{w}_{t}\\
&\mbox{where}\\
&\Phi_{t}^{i,j} = A \left( {\Omega}_{t|t}^{i} \right)^{-1} \\
&\times \left[ \epsilon^{L} \left\{ \sum_{k=1}^{L} T_{t,k-1}^{j} \prod_{m=1}^{k-1} \left( I -T_{t,m-1}^{j} \right) {\Omega}_{t|t-1}^{j} \right\} \right]  \\
&\Gamma_{t}^{i,j} = - A N \left( {\Omega}_{t|t}^{i} \right)^{-1}\\
&\times \left[ \epsilon^{L} \left\{ \sum_{k=1}^{L} T_{t,k-1}^{j} \prod_{m=1}^{k-1} \left( I -T_{t,m-1}^{j} \right) \left( {C}_{t}^{j} \right)^{T} {V}_{t}^{j} \right\} \right]
\end{split}
\end{equation}
\end{lemma}

\begin{proof}
$\bm{e}_{t+1}^{i}$ can be expressed as \eqref{eq:18} from \eqref{eq:9}.
\begin{equation}
\begin{split} \label{eq:18}
\bm{e}_{t+1}^{i} = A \left( \bm{x}_{t} - \hat{\bm{x}}_{t|t}^{i} \right) + \bm{w}_{t}
\end{split}
\end{equation}
By using the proposed algorithm in Table \ref{table_2}, the estimated state $\hat{\bm{x}}_{t|t}^{i}$ can be expressed as
\begin{equation}
\resizebox{0.45\textwidth}{!}{$
\begin{split} \label{eq:19}
&\hat{\bm{x}}_{t|t}^{i} = \left( {\Omega}_{t|t}^{i} \right)^{-1} N \bm{b}_{t}^{i} \left( L \right) \\
&= \left( {\Omega}_{t|t}^{i} \right)^{-1}
N \left[ \sum_{j\in\mathcal{N}_{i}} \epsilon^{L} \left\{ \sum_{k=1}^{L} T_{t,k-1}^{j} \prod_{m=1}^{k-1} \left( I -T_{t,m-1}^{j} \right) \bm{b}_{t}^{j} \left( 0 \right) \right\} \right] \\
&= \left( {\Omega}_{t|t}^{i} \right)^{-1} \left[ \sum_{j\in\mathcal{N}_{i}} \epsilon^{L} \left\{ \sum_{k=1}^{L} T_{t,k-1}^{j} \prod_{m=1}^{k-1} \left( I -T_{t,m-1}^{j} \right)\bm{q}_{t|t-1}^{j} \right\} \right] \\
&+ \left( {\Omega}_{t|t}^{i} \right)^{-1} N \left[ \sum_{j\in\mathcal{N}_{i}} \epsilon^{L} \left\{ \sum_{k=1}^{L} T_{t,k-1}^{j} \prod_{m=1}^{k-1} \left( I -T_{t,m-1}^{j} \right)  \left( {C}_{t}^{j} \right)^{T} {V}_{t}^{j} \bm{y}_{t}^{j} \right\} \right]
\end{split}
$}
\end{equation}
Also, $\bm{x}_{t}$ can be rewritten as \eqref{eq:20} from Table \ref{table_2}.
\begin{equation}
\resizebox{0.45\textwidth}{!}{$
\begin{split} \label{eq:20}
&{\bm{x}}_{t} = \left( {\Omega}_{t|t}^{i} \right)^{-1} N B_{t}^{i} \left( L \right) {\bm{x}}_{t} \\
&= \left( {\Omega}_{t|t}^{i} \right)^{-1}
N \left[ \sum_{j\in\mathcal{N}_{i}} \epsilon^{L} \left\{ \sum_{k=1}^{L} T_{t,k-1}^{j} \prod_{m=1}^{k-1} \left( I -T_{t,m-1}^{j} \right) B_{t}^{j} \left( 0 \right) \right\} \right] {\bm{x}}_{t} \\
&= \left( {\Omega}_{t|t}^{i} \right)^{-1} \left[ \sum_{j\in\mathcal{N}_{i}} \epsilon^{L} \left\{ \sum_{k=1}^{L} T_{t,k-1}^{j} \prod_{m=1}^{k-1} \left( I -T_{t,m-1}^{j} \right) {\Omega}_{t|t-1}^{j} \right\} \right] {\bm{x}}_{t} \\
&+ \left( {\Omega}_{t|t}^{i} \right)^{-1} N \left[ \sum_{j\in\mathcal{N}_{i}} \epsilon^{L} \left\{ \sum_{k=1}^{L} T_{t,k-1}^{j} \prod_{m=1}^{k-1} \left( I -T_{t,m-1}^{j} \right)\left( {C}_{t}^{j} \right)^{T} {V}_{t}^{j} {C}_{t}^{j} \right\} \right] {\bm{x}}_{t}
\end{split}
$}
\end{equation}
From \eqref{eq:9}, \eqref{eq:18}, \eqref{eq:19} and \eqref{eq:20}, dynamics of $\bm{e}_{t}^{i}$ is derived as
\begin{equation}
\resizebox{0.43\textwidth}{!}{$
\begin{split} \label{eq:21}
&\bm{e}_{t+1}^{i} =A \left( {\Omega}_{t|t}^{i} \right)^{-1} \\
&\times \left[ \sum_{j \in \mathcal{N}_{i}} \epsilon^{L} \left\{ \sum_{k=1}^{L} T_{t,k-1}^{j} \prod_{m=1}^{k-1} \left( I -T_{t,m-1}^{j} \right) {\Omega}_{t|t-1}^{j} \bm{e}_{t}^{j} \right\} \right]  \\
&- A \left( {\Omega}_{t|t}^{i} \right)^{-1} N\\
&\times \left[ \sum_{j \in \mathcal{N}_{i}} \epsilon^{L} \left\{ \sum_{k=1}^{L} T_{t,k-1}^{j} \prod_{m=1}^{k-1} \left( I -T_{t,m-1}^{j} \right) \left( {C}_{t}^{j} \right)^{T} {V}_{t}^{j} \bm{v}_{t}^{j} \right\} \right]\\
& + \bm{w}_{t}
\end{split}
$}
\end{equation}
which can be readily expressed as
\begin{equation}
\begin{split} \label{eq:22}
&\bm{e}_{t+1}^{i}=\sum_{j \in \mathcal{N}_{i}} \Phi_{t}^{i,j} \bm{e}_{t}^{j} + \sum_{j \in \mathcal{N}_{i}} \Gamma_{t}^{i,j} \bm{v}_{t}^{j}+ \bm{w}_{t} 
\end{split}
\end{equation}

\end{proof}

The stability of the noise-free collective estimation error dynamics is shown in Lemma \ref{lemma_3}.
\begin{lemma}
\label{lemma_3}
(Stability of the Noise-free Collective Estimation Error Dynamics)
Let the Assumption \ref{assumption_1} hold. Then, the time-varying system $\bm{e}_{t+1} = \Phi_{t} \bm{e}_{t} $ is uniformly exponentially stable.
\end{lemma}

\begin{proof}
The noise-free collective estimation error dynamics is defined as \eqref{eq:23} from \eqref{eq:22}.
\begin{equation}
\begin{split} \label{eq:23}
\bm{e}_{t+1} = \Phi_{t} \bm{e}_{t}
\end{split}
\end{equation}
$\bm{e}_{t}=col(\bm{e}_{t}^{i}, i \in \mathcal{N})$ and $\Phi_{t}$ is a block matrix, whose block elements are $\Phi_{t}^{i,j}$ in \eqref{eq:22}.

Let $\bm{p}$ represent the Perron-Frobenius left eigenvector of the matrix $\Pi^{L}$.
Under the A3 in the Assumption \ref{assumption_1}, $\bm{p}$ has strictly positive components $p^{i}$ for all $i \in \mathcal{N}$ and satisfies $\bm{p}^{T} \Pi^{L} = \bm{p}^{T}$, i.e., $\sum_{j \in \mathcal{N}} p^{j} \epsilon=p^{i}$.
Consider the following Lyapunov candidate function
\begin{equation}
\begin{split} \label{eq:24}
V_{t} (\bm{e}_{t}) =\sum_{i \in \mathcal{N}} p^{i} (\bm{e}_{t}^{i})^{T} \Omega_{t|t-1}^{i} \bm{e}_{t}^{i}
\end{split}
\end{equation}
From Lemma \ref{lemma_1}, it can be shown that there exist $\alpha_{1}>0$ and $\alpha_{2}>0$, which satisfy $\alpha_{1} \norm{\bm{e}_{t}}^{2} \leq V_{t} (\bm{e}_{t}) \leq \alpha_{2} \norm{\bm{e}_{t}}^{2}$.
Thus, $V_{t} (\bm{e}_{t}) > 0$ except $\bm{e}_{t}=\bm{0}$.

Note that there exists a positive real $\tilde{\beta}<1$ such that $\Omega_{t+1|t}^{i} \leq \tilde{\beta}A^{-T} \Omega_{t|t}^{i} A^{-1}$ from Lemma 1 of \cite{battistelli2014kullback}.
From \eqref{eq:10}, we can imply that
\begin{equation}
\Omega_{t|t}^{i} \geq \sum_{j \in \mathcal{N}^{i}} \left[ \epsilon^{L} \left\{ \sum_{k=1}^{L} T_{t,k-1}^{j} \prod_{m=1}^{k-1} \left( I -T_{t,m-1}^{j} \right) {\Omega}_{t|t-1}^{j}  \right\} \right]
\end{equation}
Then, from \eqref{eq:17} and Lemma 2 of \cite{battistelli2014kullback}, \eqref{eq:25} can be readily obtained.
\begin{equation}
\resizebox{0.43\textwidth}{!}{$
\begin{split} \label{eq:25}
&(\bm{e}_{t+1}^{i})^{T} \Omega_{t+1|t}^{i} \bm{e}_{t+1}^{i}= \left( \sum_{j \in \mathcal{N}^{i}} A \left( {\Omega}_{t|t}^{i} \right)^{-1} \right. \\
&\left. \times\left[ \epsilon^{L} \left\{ \sum_{k=1}^{L} T_{t,k-1}^{j} \prod_{m=1}^{k-1} \left( I -T_{t,m-1}^{j} \right) {\Omega}_{t|t-1}^{j} \right\} \right] \bm{e}_{t}^{j} \right)^{T} {\Omega}_{t+1|t}^{i} \\
&\times\sum_{j \in \mathcal{N}^{i}} A \left( {\Omega}_{t|t}^{i} \right)^{-1}\left[ \epsilon^{L} \left\{ \sum_{k=1}^{L} T_{t,k-1}^{j} \prod_{m=1}^{k-1} \left( I -T_{t,m-1}^{j} \right) {\Omega}_{t|t-1}^{j} \right\} \right] \bm{e}_{t}^{j} \\
&\leq \tilde{\beta} \left( \sum_{j \in \mathcal{N}^{i}} \left[ \epsilon^{L} \left\{ \sum_{k=1}^{L} T_{t,k-1}^{j} \prod_{m=1}^{k-1} \left( I -T_{t,m-1}^{j} \right) {\Omega}_{t|t-1}^{j} \right\} \right] \bm{e}_{t}^{j} \right)^{T} \\ & \left( \Omega_{t|t}^{i} \right)^{-1} \sum_{j \in \mathcal{N}^{i}} \left[ \epsilon^{L} \left\{ \sum_{k=1}^{L} T_{t,k-1}^{j} \prod_{m=1}^{k-1} \left( I -T_{t,m-1}^{j} \right) {\Omega}_{t|t-1}^{j} \right\} \right] \bm{e}_{t}^{j}\\
&\leq \tilde{\beta} \sum_{j \in \mathcal{N}^{i}} (\bm{e}_{t}^{j})^{T}  \left[ \epsilon^{L} \left\{ \sum_{k=1}^{L} T_{t,k-1}^{j} \prod_{m=1}^{k-1} \left( I -T_{t,m-1}^{j} \right) {\Omega}_{t|t-1}^{j} \right\} \right] \bm{e}_{t}^{j}\\
&\leq \tilde{\beta} \epsilon \sum_{j \in \mathcal{N}^{i}} (\bm{e}_{t}^{j})^{T}  {\Omega}_{t|t-1}^{j} \bm{e}_{t}^{j}\\
\end{split}
$}
\end{equation}
Consequently, \eqref{eq:26} can be derived from \eqref{eq:24} and \eqref{eq:25}.
\begin{equation}
\begin{split} \label{eq:26}
V_{t+1} (\bm{e}_{t+1}) &= \sum_{i \in \mathcal{N}} p^{i} (\bm{e}_{t+1}^{i})^{T} \Omega_{t+1|t}^{i} \bm{e}_{t+1}^{i}\\
&\leq \tilde{\beta} \sum_{i, j \in \mathcal{N}} p^{i} \epsilon (\bm{e}_{t}^{j})^{T} \Omega_{t|t-1}^{j} \bm{e}_{t}^{j} \\
&=\tilde{\beta}\sum_{j \in \mathcal{N}} p^{j} (\bm{e}_{t}^{j})^{T} \Omega_{t|t-1}^{j} \bm{e}_{t}^{j} \\
&=\tilde{\beta} V_{t}(\bm{e}_{t})
\end{split}
\end{equation}
\end{proof}
From Lemma \ref{lemma_1}, Lemma \ref{lemma_2} and Lemma \ref{lemma_3}, the stability of the collective estimation error dynamics with the ICF with partial information exchange is proven in Theorem \ref{theorem_1}.
\begin{theorem}
\label{theorem_1}
(Stability of the Collective Estimation Error Dynamics)
Let the Assumption \ref{assumption_1} hold. Then, the estimation error $\bm{e}_{t}^{i} = \bm{x}_{t} - \hat{\bm{x}}_{t|t-1}^{i}$ with the ICF with partial information exchange is asymptotically bounded in mean square for every $i \in \mathcal{N}$, i.e. ${\lim \ \sup}_{t \to \infty} \mathbb{E} \left\{ || \bm{e}_{t}^{i} ||^{2} \right\} < +\infty$, for any $i \in \mathcal{N}$.
\end{theorem}

\begin{proof}
From Lemma \ref{lemma_3}, it is proven that the noise-free collective estimation error dynamics is uniformly exponentially stable.
The process and the measurement noise-related terms in the error dynamics \eqref{eq:17}, $\sum_{j \in N} \Gamma_{t}^{i,j} \bm{v}_{t}^{j} + \bm{w}_{t}$, are bounded in mean square since Lemma \ref{lemma_1} holds.
Thus, Theorem \ref{theorem_1} holds. 
\end{proof}

\section{Convergence to Centralized Estimation Algorithm}
In Section V, convergence of the proposed ICF with partial information exchange to the C(E)KF is investigated.
The assumptions required for this convergence analysis are suggested in Assumption \ref{assumption_2}.
\begin{assumption}
\label{assumption_2}
(Assumptions for Convergence Analysis) \\
A1. The system matrix $A$ or $A_{t}$ is invertible. \\
A2. The system is collectively observable. \\
A3. The consensus matrix $\Pi$, whose elements are the consensus weights $\epsilon$, is row stochastic and primitive. \\
A4. Every sensor node broadcasts the elements at the same positions in $\bm{b}_{t}^{j}$ and $B_{t}^{j}$ for each consensus step ( i.e. $T_{t,l}^{j} = T_{t,l}^{i}$ for all $j$ ).
\end{assumption}

The proposed ICF with partial information exchange is proven to converge to the C(E)KF in Theorem \ref{theorem_2}.
\begin{theorem}
\label{theorem_2}
(Convergence of the ICF with Partial Information Exchange to the C(E)KF)
Let the Assumption \ref{assumption_2} holds.
Then, the ICF with partial information exchange asymptotically converges to the C(E)KF.
\end{theorem}

\begin{proof}
Let $\alpha_{t}^{i}$ represent $\bm{b}_{t}^{i}$ or $B_{t}^{i}$ in Table \ref{table_2}.
Then, the consensus step in Table \ref{table_2} can be rewritten with respect to $\alpha_{t}^{i} \left( l \right)$ as
\begin{equation}
\begin{split}
\label{eq:27}
\alpha_{t}^{i} \left( l \right) = \epsilon \sum_{j \in \mathcal{N}^{i}} \left[ T_{t, l-1}^{j} \alpha_{t}^{j} \left( l-1 \right) + \left( I - T_{t, l-1}^{j} \right) \alpha_{t}^{j} \left( l-1 \right) \right]
\end{split}
\end{equation}
Considering the A4 in the Assumption \ref{assumption_2}, \eqref{eq:28} can be derived from \eqref{eq:27} with $l=1,\cdots,L \left(=c\bar{\theta} \right)$.
\begin{equation}
\begin{split} \label{eq:28}
&\alpha_{t}^{i} \left( 1 \right) = \epsilon \sum_{j \in \mathcal{N}^{i}} \left[ T_{t, 0}^{i} \alpha_{t}^{j} \left( 0 \right) + \left( I - T_{t, 0}^{i} \right) \alpha_{t}^{i} \left( 0 \right) \right]  \\
&\quad \quad \quad \quad \vdots \\
&\alpha_{t}^{i} \left( \bar{\theta} \right) = \epsilon \sum_{j \in \mathcal{N}^{i}} \left[ T_{t, {\bar{\theta}}-1}^{i} \alpha_{t}^{j} \left( \bar{\theta}-1 \right) + \left( I - T_{t, {\bar{\theta}}-1}^{i} \right) \alpha_{t}^{i} \left( \bar{\theta}-1 \right) \right]  \\
&\alpha_{t}^{i} \left( \bar{\theta}+1 \right) = \epsilon \sum_{j \in \mathcal{N}^{i}} 
\left[ T_{t, {0}}^{i} \alpha_{t}^{j} \left( \bar{\theta} \right) + \left( I - T_{t, {0}}^{i} \right) \alpha_{t}^{i} \left( \bar{\theta} \right) \right] \\
&\quad \quad \quad \quad \vdots \\
&\alpha_{t}^{i} \left( L \right) = \epsilon \sum_{j \in \mathcal{N}^{i}} 
\left[ T_{t, \bar{\theta}-1}^{i} \alpha_{t}^{j} \left( L-1 \right) + \left( I - T_{t, \bar{\theta}-1}^{i} \right) \alpha_{t}^{i} \left( L-1 \right) \right] 
\end{split}
\end{equation}
By multiplying each of $T_{t,{0}}^{i}$, $\cdots$, and $T_{t,{\bar{\theta}-1}}^{i}$ on both sides of \eqref{eq:28}, for each $\sigma=1,\cdots,\bar{\theta}$, we have
\begin{equation}
\begin{split} \label{eq:29}
\begin{cases}
T_{t, {\sigma-1}}^{i} \alpha_{t}^{i} \left( l + 1 \right) = \epsilon \sum_{j \in \mathcal{N}^{i}} T_{t, {\sigma-1}}^{i} \alpha_{t}^{j} \left( l \right) & l \ \mbox{mod} \ \bar{\theta} = \sigma-1 \\
T_{t, {\sigma-1}}^{i} \alpha_{t}^{i} \left( l + 1 \right) = T_{t, {\sigma-1}}^{i} \alpha_{t}^{i} \left( l \right) & \mbox{otherwise}
\end{cases}
\end{split}
\end{equation}
where $l=0,\cdots,L-1$.
By summing the first equation of \eqref{eq:29} from $\sigma=1$ to $\bar{\theta}$, one can imply that
\begin{equation}
\begin{split} \label{eq:30}
\sum_{\sigma=1}^{\bar{\theta}} T_{t, {\sigma-1}}^{i} \alpha_{t}^{i} \left( l+1 \right) = \sum_{\sigma=1}^{\bar{\theta}} \epsilon T_{t, {\sigma-1}}^{i} \sum_{j \in \mathcal{N}^{i}} \alpha_{t}^{j} \left( l \right)
\end{split}
\end{equation}
If consensus is conducted with infinite steps (i.e. $l \to \infty$), \eqref{eq:30} further becomes 
\begin{equation}
\begin{split} \label{eq:31}
\sum_{\sigma=1}^{\bar{\theta}} T_{t, {\sigma-1}}^{i} \lim_{l \to \infty} \alpha_{t}^{i} \left( l+1 \right) = \sum_{\sigma=1}^{\bar{\theta}} T_{t, {\sigma-1}}^{i} \lim_{l \to \infty} \sum_{j \in \mathcal{N}^{i}} \epsilon \alpha_{t}^{j} \left( l \right)
\end{split}
\end{equation}
Since $\lim_{l \to \infty} \sum_{j \in \mathcal{N}^{i}} \epsilon \alpha_{t}^{j} \left( l \right) = \frac{1}{N} \sum_{j=1}^{N} \alpha_{t}^{j}\left( 0 \right)$ for strongly connected networks from the property of the average consensus and $\sum_{\sigma=1}^{\bar{\theta}} T_{t, {\sigma-1}}^{i} = I$, we have
\begin{equation}
\begin{split} \label{eq:32}
\lim_{l \to \infty} \alpha_{t}^{i} \left( l \right) = \frac{1}{N} \sum_{j=1}^{N} \alpha_{t}^{j}\left( 0 \right)
\end{split}
\end{equation}
From Table \ref{table_2} and \eqref{eq:32}, $\lim_{l \to \infty} {B}_{t}^{i} \left( l \right)$ and $\lim_{l \to \infty} \bm{b}_{t}^{i} \left( l \right)$ can be expressed as
\begin{equation}
\begin{split} \label{eq:33}
&\lim_{l \to \infty} {B}_{t}^{i} \left( l \right) = \frac{1}{N} \sum_{j=1}^{N} \left[ \frac{1}{N}{\Omega}_{t|t-1}^{j} + \delta {\Omega}_{t}^{j} \right] \\
&\lim_{l \to \infty} \bm{b}_{t}^{i} \left( l \right) = \frac{1}{N} \sum_{j=1}^{N} \left[ \frac{1}{N}\bm{q}_{t|t-1}^{j} + \delta \bm{q}_{t}^{j} \right]
\end{split}
\end{equation}
Consensus on priors guarantees that all local sensor nodes have the same a priori estimates, resulting in ${\Omega}_{t|t-1}=\frac{1}{N} \sum_{j=1}^{N} {\Omega}_{t|t-1}^{j}$ and $\bm{q}_{t|t-1}=\frac{1}{N} \sum_{j=1}^{N} \bm{q}_{t|t-1}^{j}$.
Thus, \eqref{eq:34} can be derived from \eqref{eq:33} and Table \ref{table_2}.
\begin{equation}
\begin{split} \label{eq:34}
&{\Omega}_{t|t} = N \lim_{l \to \infty} {B}_{t}^{i} \left( l \right) ={\Omega}_{t|t-1} + \sum_{i=1}^{N} \left( {C}_{t}^{i} \right)^{T} {V}_{t}^{i} {C}_{t}^{i} \\
&\bm{q}_{t|t} = {\Omega}_{t|t} \bm{x}_{t|t} = N \lim_{l \to \infty} \bm{b}_{t}^{i} \left( l \right) =\bm{q}_{t|t-1} + \sum_{i=1}^{N} \left( {C}_{t}^{i} \right)^{T} {V}_{t}^{i} \bm{y}_{t}^{i}
\end{split}
\end{equation}
${\Omega}_{t|t}$ and $\bm{q}_{t|t}$ in \eqref{eq:34} are identical to the information matrix and the information vector of the C(E)KF summarized in Table \ref{table_1}.
\end{proof}

\section{Simulation}
Simulations are performed to demonstrate performance of the proposed ICF with partial information exchange and verify the theoretical findings in Section IV and V.

\subsection{Simulation Settings}
In this simulation study, $100$ Monte-Carlo simulations are conducted.
An example simulation environment is illustrated in Fig. \ref{fig1}.
\begin{figure}[!h]
   \centering
   \includegraphics[width=0.45\textwidth]{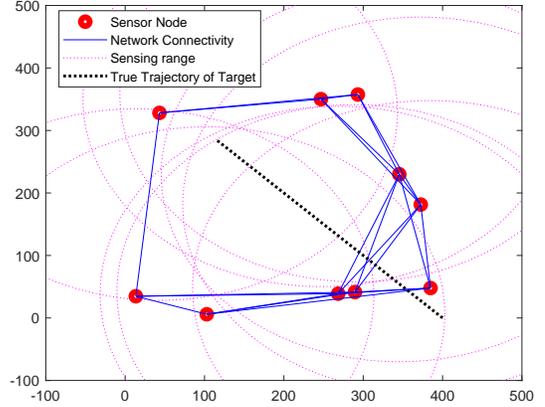}
   \caption{Example Simulation Environment}
   \label{fig1}
\end{figure}
The sensor network has $10$ sensor nodes and each node communicates with other nodes within the range of $300m$.
The positions of the sensor nodes are randomly selected in each simulation and the sensing range is $300m$.
The state vector $\bm{x}_{t}$ is defined with the $2$-dimensional position and velocity of the target (i.e. $\bm{x}_{t}=\left[ x_{T}, \ y_{T}, \ {V_{x}}_{T}, \ {V_{y}}_{T}\right]^{T}$).
The initial position of the target $\left[ x_{0,T}, \ y_{0,T} \right]^{T}$ is set to be $\left[ 400m, 0 \right]^{T}$.
The initial velocity of the target $\left[ {V_{x}}_{0,T}, \ {V_{y}}_{0,T} \right]^{T} = \left[ V_{0,T} \cos \left( \psi_{0,T} \right), V_{0,T} \sin \left( \psi_{0,T} \right) \right]^{T}$ is defined by randomly selecting $V_{0,T}$ and $\psi_{0,T}$ within $10m/s \leq V_{0,T} \leq 15m/s$ and $\frac{1}{2}\pi \leq \psi_{0,T} \leq \frac{3}{4}\pi$.
The speed of the target $V_{T}$ randomly fluctuates with the variance of $0.25 m^{2}/s^{2}$.
The system dynamics is modeled as a linear system with the system matrix $A$ in \eqref{eq:35}.
\begin{equation}
\begin{split} \label{eq:35}
A = \left[ \begin{matrix}
1 & 0 & 0.1 & 0 \\
0 & 1 & 0 & 0.1 \\
0 & 0 & 1 & 0 \\
0 & 0 & 0 & 1 \\
\end{matrix} \right] \\
\end{split}
\end{equation}
The covariance of the process noise is set to be $Q = diag([10, 10, 1 , 1])$.
The sensors measure the $2$-dimensional position of the target (i.e. $\bm{y}_{t}^{i}=\left[ x_{T}, \ y_{T}\right]^{T}$), resulting in the measurement matrix as
\begin{equation}
\begin{split} \label{eq:36}
C^{i} = \left[ \begin{matrix}
1 & 0 & 0 & 0 \\
0 & 1 & 0 & 0 \\
\end{matrix} \right] \\
\end{split}
\end{equation}
The covariance of the measurement noise is set to be $R = diag([25, 25])$.
The initial estimates of the state vector and the information matrix are set to be $\hat{\bm{x}}_{1|0}=\left[0, \ 0, \ 0, \ 0 \right]^{T}$ and ${\Omega}_{1|0}=\bm{0}_{4, 4}$, respectively.
The simulation time step is $0.1s$ and the total simulation time is $30s$.
 
To investigate the effect of the number of entries transmitted in each consensus step on the estimation performance, two different entry selection matrices in Table \ref{table_3} are utilized. Simulations are conducted by increasing the number of consensus steps at each time step from 1 to 20.

\begin{table}[!h]
\caption{Entry Selection Matrix for Simulation}
\label{table_3}
\begin{center}
\renewcommand{\arraystretch}{1.2}
\begin{tabular}{ | c | c |}
  \hline
  Case & Entry Selection Matrix \\
  \hline
  \hline
  $1$ & \makecell{ $T_{t,2k-1}^{i} = diag([1, 0, 1, 0])$, $T_{t,2k}^{i} = diag([0, 1, 0, 1])$ } \\ 
  \hline
  $2$ & \makecell{ $T_{t,4k-3}^{i} = diag([1, 0, 0, 0])$, $T_{t,4k-2}^{i} = diag([0, 1, 0, 0])$, \\ $T_{t,4k-1}^{i} = diag([0, 0, 1, 0])$, $T_{t,4k}^{i} = diag([0, 0, 0, 1])$ } \\
  \hline
\end{tabular}
\end{center}
\end{table}

\subsection{Simulation Results}

The average of the estimation error norm is plotted over time in Fig. \ref{fig2}.
In Fig. \ref{fig3}, the average of the final estimation error norm is plotted over the number of consensus step.
Note that the proposed algorithm with the entry selection matrices in Table \ref{table_3} (Case 1 and 2) is compared with the ICF and the CKF, and Fig. \ref{fig2} especially shows the simulation results with the number of consensus step set to be $4$ (Case A) and $12$ (Case B) for the consensus based approaches.

\begin{figure}[!h]
   \centering
   \includegraphics[width=0.45\textwidth]{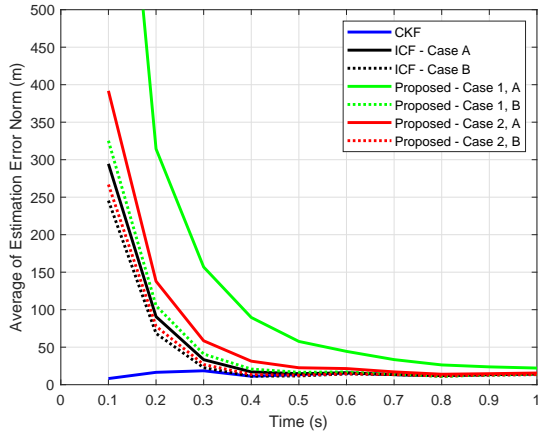}
   \caption{Average of Estimation Error Norm over Time}
   \label{fig2}
\end{figure}

\begin{figure}[!h]
   \centering
   \includegraphics[width=0.45\textwidth]{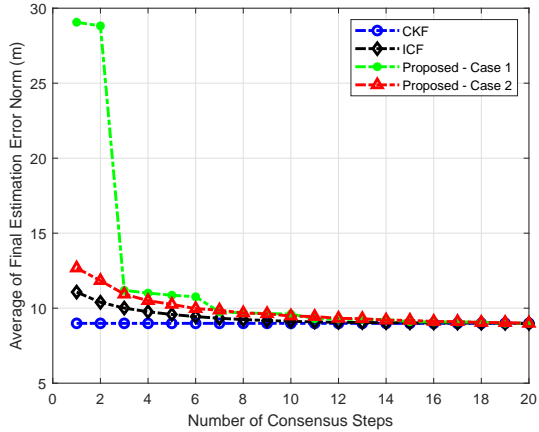}
   \caption{Average of Final Estimation Error Norm over Number of Consensus Steps}
   \label{fig3}
\end{figure}

It is observed from Fig. \ref{fig2} that the estimation error from the proposed ICF with partial information exchange is stable, and this coincides with the proof in Section IV.
Fig. \ref{fig2} shows that the estimation accuracy of the proposed algorithm becomes similar to that of the ICF and the CKF over time, and the convergence speed with the proposed algorithm is slightly lower than that of the ICF and the CKF.
It can be seen from Fig. \ref{fig2} that the speed degeneration of the proposed algorithm can be reduced as the number of the consensus steps increases.
Although there exists a trade-off between the bandwidth and the convergence speed, the bandwidth is significantly reduced with the proposed algorithm comparing to the ICF while the convergence speed is slightly degenerated: the amount of information transmitted between the sensor nodes at each consensus step is reduced by $50\%$ in Case 1 and $75\%$ in Case 2 while the settling time is increased by less than $0.7s$ in all the cases.
Fig. \ref{fig2} also indicates that the entry selection matrix influences the convergence speed and hence it would be good to optimize this matrix by considering bandwidth limit.
Fig. \ref{fig3} shows that the proposed ICF with partial information exchange converges to the CKF as the number of the consensus steps increases, and this corresponds to the proof on the guaranteed convergence to the centralized algorithm in Section V.

\section{Conclusion}
The ICF with partial information exchange which achieves both bandwidth reduction and guaranteed convergence to the centralized algorithm is successfully proposed in this paper.
In the proposed algorithm, consensus averaging is performed at each sensor node with information from the sensor nodes in the neighborhood, and information transmitted between the sensor nodes at each consensus step is chosen with the entry selection matrix.
Thus, the proposed ICF with partial information exchange is guaranteed to converge to the CKF while reducing the bandwidth of the signals transmitted between sensors.
The estimation error with the proposed algorithm is proven to be asymptotically bounded in mean square.
Also, it is proven that the proposed ICF with partial information exchange converges to the CKF by utilizing the property of the average consensus.
Numerical simulation validates the proposed ICF with partial information exchange and the related theoretical findings.

\bibliographystyle{IEEEtran}
\bibliography{Reference}

\end{document}